\documentclass[11pt]{scrartcl}

\usepackage[margin=.96in, letterpaper]{geometry}
\usepackage{times}

\usepackage{xcolor}

\definecolor{darkgreen}{rgb}{0,0.5,0}
\definecolor{darkblue}{rgb}{0,0,0.8}

\usepackage{listings} 

\usepackage{graphicx}

\usepackage{ifthen}

\usepackage{amsmath,amssymb,amsthm,amscd}
\usepackage{shadethm}
\usepackage{bbm}  													
\usepackage{tikz}\usetikzlibrary{arrows,calc,patterns,decorations.pathreplacing}
\usepackage{epic}
\usepackage{paralist}
\usepackage{todonotes}

\usepackage[all]{xy}
\allowdisplaybreaks
\usepackage{microtype}

\DeclareMicrotypeSet*[tracking]{my}%
  { font = */*/*/sc/* }%
\SetTracking{encoding = *, shape = sc }{ 45 }
\newcommand{\NH}{\mathcal{N}^{\LOC}} 
\newcommand{\NCR}{\mathcal{N}^{\SL}}
\newcommand{\NT}{\mathcal{N}} 
\newcommand{\NB}{\widetilde{\mathcal{N}}}
\newcommand{\SSS}{\mathcal{S}}

\newcommand{\TG}{\mathcal{\tilde{H}}}

\newcommand{\set}[1]{\left\{#1\right\}}

\newcommand{\SETLOCAL}{\ensuremath{\mathsf{SET}\text{-}\mathsf{LOCAL}}}
\newcommand{\SL}{\ensuremath{\mathsf{SL}}}
\newcommand{\LOCAL}{\ensuremath{\mathsf{LOCAL}}}
\newcommand{\LOC}{\ensuremath{\mathsf{LOC}}}

\newcommand{\eps}{\ensuremath{\varepsilon}}

 \definecolor{shadethmcolor}{rgb}{.94,.94,.94}   
    \definecolor{shaderulecolor}{rgb}{0.0,0.0,0.0}   

    \newshadetheorem{definitions}{Definition}[section]
    \newshadetheorem{thms}[definitions]{Theorem}
    \newshadetheorem{cors}[definitions]{Corollary}
    \newshadetheorem{props}[definitions]{Proposition}
    \newshadetheorem{lems}[definitions]{Lemma}
		\newshadetheorem{conjecture}[definitions]{Conjecture}
		\newshadetheorem{example}[definitions]{Example}
		\newshadetheorem{remark}[definitions]{Remark}
		\newshadetheorem{komm}{Komm}

    \newtheoremstyle{TheoremNum}
        {\topsep}{\topsep}              
        {\itshape}                      
        {}                              
        {\bfseries}                     
        {.}                             
        { }                             
        {\thmname{#1}\thmnote{ \bfseries #3}}
    \theoremstyle{TheoremNum}

    \newenvironment{thm}[1][]{%
      \begin{thms}[#1]\hspace*{1mm}%
    }{\end{thms}}

    \newenvironment{lemma}[1][]{%
      \begin{lems}[#1]\hspace*{1mm}%
    }{\end{lems}}

    \newenvironment{corr}[1][]{%
      \begin{cors}[#1]\hspace*{1mm}%
    }{\end{cors}}

\theoremstyle{definition}
\swapnumbers

\renewcommand{\phi}{\varphi}

\newcommand{\dcup}{\ensuremath{\mathaccent\cdot\cup}}

\newcommand{\comment}[1]{}

\newcommand{\kommentar}[1]{}

\renewcommand{\epsilon}{\varepsilon}
\renewcommand{\theta}{\vartheta}
\renewcommand{\tilde}{\widetilde}

\definecolor{light-gray}{gray}{0.60}

\newcommand{\bigO}{\mathcal{O}}
\newcommand{\logstar}{\log^{*}}

\usepackage{chngcntr}

\graphicspath{{pics/}}

\title{Polynomial Lower Bound for Distributed Graph Coloring in a Weak LOCAL Model}

\author{
  \normalsize Dan Hefetz\\
  \normalsize Ariel University, Israel\\
  \small\texttt{danny.hefetz@gmail.com}\\
  \\
  \normalsize Yannic Maus\\
  \normalsize University of Freiburg, Germany\\
  \small \texttt{yannic.maus@cs.uni-freiburg.de}
  \and
  \normalsize Fabian Kuhn\\
  \normalsize University of Freiburg, Germany\\
  \small \texttt{kuhn@cs.uni-freiburg.de}\\
  \\
  \normalsize Angelika Steger\\
  \normalsize ETH Zurich, Switzerland\\
  \small \texttt{steger@inf.ethz.ch}
}

\date{}

\begin{document}

\maketitle
\thispagestyle{empty}

\begin{abstract}
  \vspace*{-6mm}
  \noindent\textbf{Abstract:}
We show an $\Omega\big(\Delta^{\frac{1}{3}-\frac{\eta}{3}}\big)$ lower bound  on the runtime of any deterministic distributed $\bigO\big(\Delta^{1+\eta}\big)$-graph coloring algorithm in a weak variant of the \LOCAL\ model.

In particular, given a network graph \mbox{$G=(V,E)$}, in the weak \LOCAL\ model nodes communicate in synchronous rounds and they can use unbounded local computation. We assume that the nodes have no identifiers, but that instead, the computation starts with an initial valid vertex coloring. A node can \textbf{broadcast} a \textbf{single} message of \textbf{unbounded} size to its neighbors and receives the \textbf{set of messages} sent to it by its neighbors. That is, if two neighbors of a node $v\in V$ send the same message to $v$, $v$ will receive this message only a single time; without any further knowledge, $v$ cannot know whether a received message was sent by only one or more than one neighbor.

Neighborhood graphs have been essential in the proof of lower bounds for distributed coloring algorithms, e.g., \cite{linial92,Kuhn2006On}.  Our proof analyzes the recursive structure of the neighborhood graph of the respective model to devise an $\Omega\big(\Delta^{\frac{1}{3}-\frac{\eta}{3}}\big)$ lower bound on the runtime for any deterministic distributed $\bigO\big(\Delta^{1+\eta}\big)$-graph coloring algorithm.

Furthermore, we hope that the proof technique improves the understanding of neighborhood graphs in general and that it will help towards finding a lower (runtime) bound for distributed graph coloring in the standard \LOCAL\ model.
Our proof technique works for one-round algorithms in the standard \LOCAL\ model and provides a simpler and more intuitive proof for an existing $\Omega(\Delta^2)$ lower bound proven in \cite{stacs09}. This proof also extends to one-round $d$-defective coloring algorithms. We show that any one round $d$-defective color reduction algorithm in the standard \LOCAL\ model needs $\Omega(\Delta^2/(d+1)^2)$ colors if the input graph is colored with sufficiently many colors.
\end{abstract}

\setcounter{page}{0}

\clearpage
\section{Introduction}
\label{sec:intro}

\vspace*{-2mm}
In the distributed message passing model, an $n$-node communication network is represented as a graph \mbox{$G=(V,E)$}. Each node hosts a processor and processors communicate through the edges of $G$. In the standard \LOCAL\ model, time is divided into synchronous rounds and in each round, simultaneously, each node $v\in V$ performs an unbounded amount of local computations, sends a single message of unbounded size to each of its neighbors and receives the messages sent to it by its neighbors.  The time complexity of an algorithm is measured by the total number of rounds.

This paper deals with lower bounds on the time complexity of distributed graph coloring algorithms. A $c$-(vertex)-coloring of a graph $G=(V,E)$ is a function $\phi: V\rightarrow \{1,\ldots,c\}$ such that $\phi(u)\neq \phi(v)$ for all $\{u,v\}\in E$. Coloring a graph with the minimum number of colors is one of Karp's 21 NP-complete problems \cite{Karp1972} and the problem is even hard to approximate within a factor $n^{1-\eps}$ for any constant $\eps>0$ \cite{Zuckerman07}. A simple centralized greedy coloring algorithm which sequentially colors the nodes with the smallest available color guarantees a coloring with at most $\Delta+1$ colors, where $\Delta$ denotes the maximum degree of the graph. In the distributed setting, one is usually interested in competing with this greedy algorithm and to therefore find a coloring with $\Delta+1$ or more colors \cite{barenboimelkin_book}.

In this paper we consider deterministic \emph{color reduction algorithms}, where before the start of the algorithm the graph is equipped with an $m$-coloring $\psi: V\rightarrow \{1,\ldots,m\}$ (usually $m\gg\Delta$). Apart from their initial color, nodes are indistinguishable and therefore, in particular, nodes do not have unique IDs.  However, unique IDs in the range $\{1,\ldots, N\}$ for some $N\geq n$ are a special instance of the problem because they form an $N$-coloring. All recent deterministic coloring algorithms (e.g., \cite{barenboimelkin_book,barenboim15,fraigniaud15}) begin with the seminal $\bigO(\logstar n)$-round algorithm by Linial which computes an $\bigO(\Delta^2)$-coloring\footnote{The function $\logstar x$ denotes the number of iterated logarithms needed to obtain a value at most $1$, that is,\\ $\forall x\leq 1: \logstar x=0,\ \forall x>1: \logstar x = 1 + \logstar\log x$.} \cite{linial92}. Afterwards none of the algorithms make use of the unique IDs again and even Linial's algorithm does not require unique IDs but only an initial coloring of the nodes, i.e., the algorithms fit in the framework of color reduction algorithms. A lower bound for color reduction algorithms is thus almost as relevant as a lower bound for unique IDs.

More specifically, we consider color reduction algorithms in a weak variant of the standard \LOCAL\ model, which we name the \SETLOCAL\ model. In each round, each node can send an arbitrarily large message to its neighbors. However, instead of receiving one message from each neighbor, each node only receives the set of messages sent to it by its neighbors. 
That is, if two or more neighbors send the same message to a node $u$, $u$ only receives this message once.\footnote{A similar model, but for completely anonymous graphs, has previously been studied in \cite{hella15}.} Note that when assuming unique IDs, there is no difference in power between the \SETLOCAL\ model and the standard \LOCAL\ model. Every node can just add its ID to all its messages and each node can then always easily distinguish between the messages sent to it by different neighbors. However when considering color reduction algorithms, neighbors with the same inital color might send the same message even when including their color or any other local knowledge in their messages.

\vspace*{-1mm}
\paragraph{Contributions:}
As our main result, we prove a polynomial (in the maximum degree $\Delta$) lower bound on the time required by color reduction algorithms in the \SETLOCAL\ model (for a formal definition of color reduction and of the \SETLOCAL\ model, see Section \ref{sec:model}). Formally, we prove the following main theorem.

\vspace*{-1mm}
\begin{thm}[Color Reduction Lower Bound]\label{thm:lowerbound}
  Let $0\leq \eta < 1$ and $C>0$ be two constants and assume that $m\geq 2C\Delta^{1+\eta}$. Any deterministic color reduction algorithm which, given an initial $m$-coloring, computes a coloring with at most $C\Delta^{1+\eta}$ colors in graphs with maximum degree at most $\Delta$ in the \SETLOCAL\ model requires $\Omega\big(\Delta^{\frac{1-\eta}{3}}\big)$ rounds.

\smallskip

\noindent Thus, in particular, any color reduction algorithm for computing a $(\Delta+1)$-coloring needs $\Omega\big(\Delta^{\frac{1}{3}}\big)$ rounds.
\end{thm}

Note that the theorem in particular implies that the time required for computing a $\Delta^{2-\eps}$-coloring for any constant $\eps>0$ is at least polynomial in $\Delta$ when using color reduction algorithms in the \SETLOCAL\ model. 

In order to establish that there are non-trivial color reduction algorithms in the \SETLOCAL\ model, we also show that an existing distributed coloring algorithm from \cite{Kuhn2006On} works in this setting. The discussion of the following theorem appears in Section \ref{sec:UpperBounds}.

\vspace*{-1mm}
\begin{thm}[Color Reduction Upper Bound]\label{thm:upperbound}
  In graphs with maximum degree at most $\Delta$ and an initial $m$-coloring, there is a deterministic distributed color reduction algorithm in the \SETLOCAL\ model which computes a $(\Delta+1)$-coloring in $\bigO(\Delta\log\Delta + \logstar m)$ rounds.
\end{thm}

Furthermore, we hope that the proof technique of the lower bound improves the understanding of neighborhood
graphs in general and that it will help towards finding a lower (runtime) bound for distributed graph
coloring in the standard \LOCAL\ model. E.g., the same proof technique works for one-round coloring algorithms in the standard \LOCAL\ model and provides a simpler and more intuitive proof for the existing lower bound in \cite{stacs09}. A slight modification of the proof yields a previously unknown lower bound for one-round $d$-defective coloring algorithms. In a $d$-defective coloring each color class induces a graph with maximum degree $d$ and we obtain the following theorem which is proven in Section \ref{sec:LowerBoundWeakColoring}.
\begin{thm}\label{thm:DefectiveoneRound}
  For all $\Delta\geq 2, d\geq 0$ any one-round $d$-defective color reduction algorithm in the standard \LOCAL\ model needs $\Omega(\Delta^2/(d+1)^2)$ colors if $m\geq 2\Delta^2$.
\end{thm}

\vspace*{-5mm}
\paragraph{Related Work:}
Distributed coloring has been identified as one of the prototypical problems to understand the problem of breaking symmetries in distributed and parallel systems. 
In the following, we discuss the work which is most relavant in the context of this paper. For a more general overview of the research on distribted coloring, we refer to the monograph of Barenboim and Elkin \cite{barenboimelkin_book}.

In a classic paper, Cole and Vishkin showed that a ring network can be $3$-colored in $\bigO(\logstar N)$ synchronous rounds, where $N$ is the size of the space of possible node IDs \cite{cole86}.\footnote{The algorithm of \cite{cole86} was described as a PRAM algorithm, however it can be directly applied in the usual distributed setting.} The algorithm was generalized in \cite{goldberg88} to a distributed $\bigO(\Delta^2 + \logstar N)$-round algorithm for $(\Delta+1)$-coloring graphs with maximum degree at most $\Delta$. Most relevant in the context of this work is the seminal paper by Linial \cite{linial92}, where he in particular shows that the $\bigO(\logstar N)$ algorithm of \cite{cole86} is asymptotically optimal and that in $\bigO(\logstar N)$ rounds, it is possible to color arbitrary graphs with $\bigO(\Delta^2)$ colors. While there has been a lot of progress on developing distributed coloring algorithms, the lower bound of \cite{linial92} is still the best known time lower bound for the standard distributed coloring problem. All the above algorithms are deterministic and at the core, they are all based on iterative color reduction schemes where a given valid vertex coloring is improved in a round-by-round manner. A different approach is taken  in \cite{awerbuch89,panconesi95}, where it is shown how to compute a $(\Delta+1)$-coloring in $2^{\bigO(\sqrt{\log n})}$ rounds ($n$ is the number of nodes) based on first computing a decomposition of the network into clusters of small diameter. When measuring the time as a function of $n$, this is still the best known deterministic distributed $(\Delta+1)$-coloring algorithm for general graphs.

There has been significant recent progress on developing faster deterministic distributed coloring algorithms, particularly for graphs with moderately small maximum degree $\Delta$. In \cite{Kuhn2006On}, it was shown that combined with a simple interative color reduction scheme, the algorithm of \cite{linial92} can be turned into a $\bigO(\Delta\log\Delta + \logstar N)$-time $(\Delta+1)$-coloring algorithm. By decomposing a graph into subgraphs with small maximum degree, an improved time complexity of $\bigO(\Delta + \logstar N)$ was achieved in \cite{BEK15}. The basic ideas of \cite{BEK15} were extended and generalized in \cite{barenboim10}, where in particular it was shown that an $\bigO(\Delta^{1+o(1)})$-coloring can be computed in time $\bigO(\Delta^{o(1)} + \logstar N)$. The time complexity for $(\Delta+1)$-colorings was recently improved in \cite{barenboim15} and \cite{fraigniaud15}, where upper bounds of $\tilde{\bigO}(\Delta^{3/4}) + \logstar N$ and $\tilde{\bigO}(\sqrt{\Delta}) + \logstar N$ rounds were shown. Both algorithm also work for the more general list coloring problem.\footnote{The algorithm of \cite{fraigniaud15} works for an even more general conflict coloring problem and $\tilde{\bigO}$ ignores polylog factors in $\log \Delta$.}

 While the best deterministic algorithms for distributed $(\Delta+1)$-coloring have time complexities which are polynomial in $\Delta$ or exponential in $\sqrt{\log n}$, much faster randomized algorithms are known. Based on the distributed maximal independent set algorithm of \cite{alon86,luby86} and a reduction described in \cite{linial92}, 
by using randomizatiion, a $(\Delta+1)$-coloring can be computed in $\bigO(\log n)$ rounds. This old result has recently been improved in \cite{barenboim12}, where it was shown that a $(\Delta+1)$-coloring can be computed in time $\bigO(\log\Delta) + 2^{\bigO(\sqrt{\log\log n})}$ and in \cite{hsinhao_coloring}, where the current best time bound of $\bigO(\sqrt{\log\Delta}) + 2^{\bigO(\sqrt{\log\log n})}$ was proven. Closing or understanding the gap between the distributed complexities of randomized and deterministic algorithms for $(\Delta+1)$-coloring and other basic symmetry breaking tasks is one of the main open problems in the area of distributed graph algorithms. Even though we are dealing with a weaker, non-standard communication model, we hope that the lower bound of the present paper provides a step in this direction. Note that for $\Delta$-coloring trees with max.\ degree at most $\Delta$, an exponential separation between randomized and deterministic algorithms has recently been shown in \cite{chang16}.

Although there has been steady progress on developing upper bounds for distributed coloring, much less is known about lower bounds. While by now there exist many distributed time lower bounds for related graph problems in the \LOCAL\ model (e.g., \cite{LLL_lowerbound,goeoes14_DISTCOMP,goeoes14_PODC,lowerbound,kuhn16_jacm}), the $\Omega(\logstar n)$ lower bound for coloring rings with $\bigO(1)$ colors by Linial \cite{linial92} is still the only time lower bound for the standard distributed coloring problem. Linial's lower bound is based on the fundamental insight that for a given $r\geq 1$, the minimum number of colors which any $r$-round coloring algorithm needs to use can be expressed as the chromatic number of a graph Linial names the \emph{neighborhood graph}. Linial then shows that the chromatic number of the $r$-round neighborhood graph for $n$-node rings is $\Omega(\log^{(2r)}n)$, where $\log^k x$ is the $k$-times iterated $\log$-function. For a more detailed discussion of how to use neighborhood graphs for proving distributed coloring lower bounds, we refer to Section \ref{sec:neighborhoodgraphs}. Using neighborhood graphs, a combination of techniques of \cite{linial92} and \cite{alon10}  also shows that coloring $d$-regular trees with less than $o(\log d /\log\log d)$ colors requires $\Omega(\log d / \log\log d)$ rounds; \cite{Barenboim2010} uses this result to show that $\bigO(a)$-coloring graphs with arboricity $a$ takes $\Omega(\log(n)/\log(a))$ rounds. Further, in \cite{Kuhn2006On,stacs09}, neighborhood graphs were used to show that in a single round, when starting with an $m$-coloring with $m$ sufficiently large, in graphs with maximum degree at most $\Delta$, the number of colors cannot be reduced to fewer than $\Omega(\Delta^2)$ colors. Similar, slightly weaker results were before already proven in \cite{szegedy93}. In \cite{LLL_lowerbound}, it has been shown that coloring $d$-regular graphs with $d$ colors requires at least $\Omega(\log\log n)$ rounds. In addition, in \cite{greedycoloring}, it was shown that $\Omega(\log n / \log\log n)$ rounds are needed to compute a 
$(\Delta+1)$-coloring where in the end, each node has the smallest possible color which is consistent with the colors chosen by its neighbors.

\vspace*{-2mm}
\section{Model \& Problem Statement}
\label{sec:model}

\vspace*{-1mm}
\paragraph{Mathematical Notation:}
 For a graph $G=(V,E)$ and a node $v\in V$, $\Gamma_G(v)$ denotes the set of neighbors of $v$ in $G$. Sometimes we write $\Gamma(v)$ if the graph $G$ is clear from the context. Given a graph $G$, we use $\Delta(G)$ to denote the maximum degree of $G$ and $\chi(G)$ to denote the chromatic number of $G$ (i.e., the number of colors of a minimum valid vertex coloring). We sometimes abuse notation and identify a set of nodes $S$ of $G$ with the subgraph induced by $S$. For example, we might write $S\subseteq G$, where $S$ denotes a subset of the nodes of $G$ and also the subgraph induced by $S$. By $[m]$ we denote the set of integers $\{1,\ldots,m\}$.

\vspace*{-2mm}
\paragraph{The Color Reduction Problem:} In the distributed color reduction problem, we are given a network graph  $G=(V,E)$ of max.\ degree at most $\Delta$. Each node $v\in V$ is equipped with an initial color $\psi(v)\in [m]$ such that the coloring $\psi$ provides a valid vertex coloring of $G$. At the start, nodes can only be distinguished by their initial color and thus at the beginning, except for the value of their initial color, all nodes start in the same state. The goal of a color reduction algorithm is to compute a new color $\phi(v)$ for each $v\in V$ such that the coloring $\phi$ also provides a valid vertex coloring of $G$, but such that the colors $\phi(v)$ are from a much smaller range. We say that a color reduction algorithm computes a $c$-coloring of $G$ if $\phi(v)\in [c]$ for all $v\in V$.

\vspace*{-2mm}
\paragraph{Communication Model:}
 We work with an adapted version of the \LOCAL\ model \cite{linial92,peleg00}, which we call the \SETLOCAL\ model.  A communication network is modeled as an $n$-node graph \mbox{$G=(V,E)$}, where the nodes of $G$ can use unbounded local computation and communicate through the edges of $G$ in synchronized rounds. In each round, a node can \textbf{broadcast} a \textbf{single} message of \textbf{unbounded} size to its neighbors and each node receives the \textbf{set of messages} sent to it by its neighbors. That is, if two or more neighbors of a node $v\in V$ send the same message to $v$, $v$ will receive the message only a single time. Thus, without any further knowledge, $v$ cannot know whether a message was sent by only one or more than one neighbor. Note that a node which broadcasts a single message of arbitrary size can send different messages to different neighbors by indicating which part of the message is for which neighbor. However, to do so it is necessary that the node can already distinguish its neighbors by some property, e.g., by the use of (different) messages received from them previously. In \cite{hella15}, different weak variants of the \LOCAL\ model were studied for problems where the network nodes are completely anonymous without any initial labeling. The \SETLOCAL\ model corresponds to the $\mathsf{SB}$ model in the hierarchy of models discussed in \cite{hella15}.

When running a distributed color reduction algorithm in the \SETLOCAL\ model, we assume that all nodes are aware of the parameters $m$ and $\Delta$ and of the number of nodes $n$ of $G$. Note that since our main focus is proving a lower bound, this assumption only makes the results stronger.

\vspace*{-2mm}
\paragraph{The Role of Randomness:}
Generally, there is a large gap between the best known randomized and deterministic distributed coloring algorithms and understanding whether this large gap is inherent or to what extent it can be closed is one of the major open problems in the area of distributed graph algorithms. When considering color reduction algorithms as introduced above, randomness can only help if either an upper bound on $n$ is known or if the running time can depend on $n$. To see this, assume that we have a randomized color reduction algorithm which computes a $c(m,\Delta)$-coloring in $T(m,\Delta)$ rounds. To have an algorithm which cannot be derandomized trivially, the algorithm must either fail to terminate in $T(m,\Delta)$ rounds with positive probability $\eps>0$ or it must fail to compute a valid $c$-coloring with positive probability $\eps>0$. Let $G$ be a graph on which the algorithm fails in one of the two ways with a positive probability $\eps>0$. Consider a graph $H_k$ which consists of $k$ identical disjoint copies of $G$. As the randomness in the $k$ copies has to be independent, when running the algorithm, one of the $k$ copies fails with probability at least $1-(1-\eps)^k$. Note that the parameters $m$ and $\Delta$ are the same for the two graphs $G$ and $H_k$. For sufficiently large $k$, this failure probability becomes arbitrarily close to $1$.

\section{Neighborhood Graphs for Lower Bounds in Distributed Coloring}
\label{sec:neighborhoodgraphs}

Neighborhood graphs were introduced by Nati Linial in his seminal paper \cite{linial92} in which he uses them to derive his famous $\Omega(\logstar n)$ lower bound for $3$-coloring rings. Let us quickly recall his main ideas: In the \LOCAL\ model  there is no loss of generality if one assumes that an $r$-round algorithm first collects all data, which it can learn in $r$ rounds, and only then decides on its output. The data, which a synchronous $r$-round distributed algorithm running at a node $v$ can learn in this model, consists of the IDs and the topology of all nodes in distance at most $r$, except for edges between nodes in distance exactly $r$. This is called the $r$-view of a node and corresponds exactly to the knowledge a node obtains if every node forwards everything it knows (i.e., its current state) to all neighbors in every round, which it can do due to unbounded message size. If the number of IDs $n$, the maximum degree $\Delta$, and the number of rounds $r$ are fixed, there are finitely many $r$-views and an $r$-round $c$-coloring algorithm is a function from those $r$-views to $[c]$. Neighborhood graphs formalize the neighborhood relation between $r$-views. The neighborhood graph $\NH_r(n,\Delta)$ for the \LOCAL\ model has a node for each feasible $r$-view and there is an edge between two such nodes if the corresponding $r$-views can occur at neighboring nodes in some graph with $n$ nodes and maximum degree $\Delta$.

Neighborhood graphs are extensively useful when studying distributed graph coloring because any (correct) $r$-round $c$-coloring algorithm yields a $c$-coloring of the $r$-round neighborhood graph $\NH_r$ and vice versa, \cite{linial92,Kuhn2006On}. Therefore the existence of an $r$-round $c$-coloring algorithm reduces to the question whether the chromatic number of $\NH_r(n,\Delta)$ is smaller than or equal to $c$. Particularly, Linial showed   \mbox{$\chi\big(\NH_r{(n,2)}\big)\in\Omega\big(\log^{(2r)}n\big)$}  which yields his lower bound of $\Omega(\logstar n)$ rounds.


\subsection{Neighborhood Graphs in the \boldmath\SETLOCAL\ Model}

In the same way as in the \LOCAL\ model we obtain the data a node $v$ can learn in an $r$-round algorithm of the \SETLOCAL\ model if every node forwards its knowledge to all neighbors in every round. After $0$ rounds a node knows nothing but its own color, after one round it knows its own color and the \textbf{set of colors} of its neighbors, and so on. Definition \ref{def:rneighborhood} formalizes the data which a node can learn in $r$ rounds in the \SETLOCAL\ model. A node cannot detect cycles unless unique IDs are given (this holds in the \SETLOCAL\ model and in the standard \LOCAL\ model). The $r$-views are thus not formed by the actual topology of the neighborhood, but by the tree unfolding of the neighborhood. Thus for color reduction algorithms, w.l.o.g., we can restrict our attention to the case of trees. Consequently, all $r$-views can be considered as trees and we therefore define $r$-neighborhoods in the following way.

\begin{definitions}[$r$-Neighborhood]\label{def:rneighborhood}
Let $G=(V,E)$ be a tree with maximum degree at most $\Delta$ and an initial $m$-coloring $\psi:V\rightarrow \{1,\ldots,m\}$. We define
\vspace*{-1.5mm}
\begin{align*}
 \SSS^G_0(v)& :=\psi(v) & &\text{0-round view}.\\
\SSS^G_{r+1}(v) & :=\left(\SSS^G_r(v), \left\{\SSS^G_r(u) \mid u\in \Gamma_G(v)\right\}\right) & &\text{$(r+1)$-round view},
\end{align*}
where $r\geq 0$ and $\SSS^G_r(v)$ equals the data which a node $v\in V$ can learn in an $r$-round distributed algorithm (\emph{$r$-view of $v$}) in the \SETLOCAL\ model.
\end{definitions}

The $r$-view $\SSS^G_r(v)$ depends on the tree $G$, the coloring $\psi$, and the node $v$. However if we fix the number of initial colors $m$, the maximum possible degree $\Delta$, and the number of rounds $r$, the number of feasible $r$-views which can occur at any node in any tree with maximum degree $\Delta$ and any initial $m$-coloring is finite. The following definition adapts the neighborhood graphs of the \LOCAL\ model to the \SETLOCAL\ model. 
\begin{figure}
\vspace{-0.6cm}
		\centering
		\includegraphics[width=0.36\textwidth]{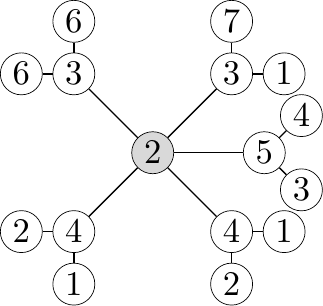}\hspace{1.5cm}
			\includegraphics[width=0.155\textwidth]{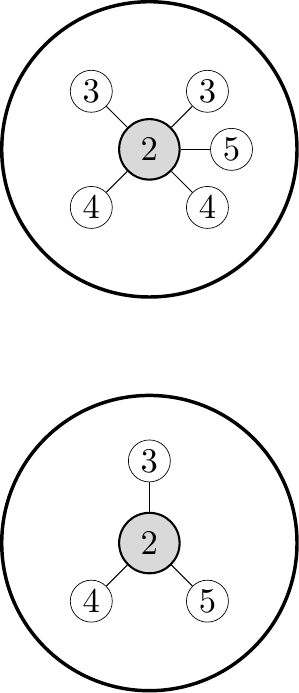}		
						\caption{\textbf{Left:} An extract of a tree graph. Images in Figure \ref{fig:neighborhoodGraphs2} correspond to views of the gray node for two rounds in different models. \textbf{Right top:} The 1-round view in the \SETLOCAL\ model of the gray node  in the left image.\newline
			\textbf{Right bottom:} The 1-round view in the \LOCAL\ model of the gray node  in the left image. 
						In  \mbox{\cite[Lemma 3.1]{Kuhn2006On} the authors prove that $\chi(\NH_1(m,D))=\chi(\NB_1(m,D))$ holds.}}	
\label{fig:setModelTree}
\end{figure}

\begin{definitions}[Neighborhood Graph in the \SETLOCAL\ Model]
Let $m$, $\Delta$,  and $r\geq 0$  be fixed integers.
Consider the following finite graph $\NCR_r(m,\Delta)$
\vspace*{-2mm}
\begin{eqnarray*}
V(\NCR_r) & := & \left\{\SSS_r^G(v) \,|\, \exists\ m\text{-colored tree $G$ s.t. } \Delta(G)\leq\Delta \text{ and } v\in V(G)  \right\},\\
E(\NCR_r) & := & \left\{\{\SSS_r^G(v),\SSS_r^G(u)\} \,|\, \exists\ m\text{-colored tree $G$ s.t. } \Delta(G)\leq\Delta \text{ and } \{u,v\}\in E(G)\right\}.
\end{eqnarray*}
\end{definitions}
Just as in the \LOCAL\ model, any $r$-round $c$-coloring algorithm in the \SETLOCAL\ model  with an initial $m$-coloring can be transformed into an equivalent algorithm in which every node first collects its $r$-neighborhood and then decides on its output. Such an algorithm is a function
 $\phi: V(\NCR_r)\rightarrow \{1,\ldots,c\}$ such that $\phi(x)\neq \phi(y)$ for all $\{x,y\}\in E(\NCR_r)$, that is, $\phi$ is a $c$-coloring of the graph $\NCR_r$.
\begin{lemma}
\label{lemma:algCorrespondence}
Any deterministic $r$-round distributed algorithm in the \SETLOCAL\ model, which correctly $c$-colors any intially $m$-colored graph with maximum degree $\Delta$, yields a feasible $c$-coloring of $\NCR_r(m,\Delta)$ and vice versa.
\end{lemma}
\begin{proof}
The following proof uses the same ideas as the proof of \cite[Lemma 3.1]{Kuhn2006On}.

Assume that we are given an $r$-round distributed algorithm in the \SETLOCAL\ model, which correctly $c$-colors any initially $m$-colored graph with maximum degree $\Delta$. This algorithm induces a $c$-coloring $\phi$ of the nodes of $\NCR_r(m,\Delta)$. Assume that this $c$-coloring is not proper, i.e., there are two nodes $x,y\in V(\NCR_r(m,\Delta))$ with $\phi(x)=\phi(y)$. By the definition of $\NCR_r(m,\Delta)$ we can construct an initially $m$-colored tree $G$ with maximum degree $\Delta$ where the $r$-round views $x$ and $y$ occur as $r$-round views of two adjacent nodes $u$ and $v$ of $G$. As a consequence the algorithm corresponding to the function $\phi$ assigns the same value to $u$ and $v$, which is a contradiction to its correctness.

For the other direction we need to prove that any $r$-round $c$-coloring of $\NCR_r(m,\Delta)$ implies an $r$-round $c$-coloring algorithm of any initially $m$-colored graph $G$ with maximum degree $\Delta$. For that purpose assume that we have a $c$-coloring $\phi$ of $\NCR_r(m,\Delta)$ which is known by all nodes of $G$. This is no problem because the graph $\NCR_r(m,\Delta)$ does only depend on $m$ and $\Delta$ and is independent of the structure of any particular graph $G$. Let each node $v$ collect its $r$-round view $\SSS_r^G(v)$ in $G$ in the \SETLOCAL\ model. By its definition $\NCR_r(m,\Delta)$ contains a node $x_v$ corresponding to the view of $v$; let $v$ select the color $\phi(x_v)$. If $u$ and $v$ are neighbors in $G$ their selected colors are different by the definition of $\NCR_r(m,\Delta)$.
\end{proof}
Due to Lemma \ref{lemma:algCorrespondence} a lower bound $\chi(\NCR_r(m,\Delta)>c$ on the chromatic number implies that there is no $r$-round color reduction algorithm in the \SETLOCAL\ model which can (correctly) $c$-color all initially $m$-colored graphs with maximum degree $\Delta$. 

\section{Lower Bound Proofs}
\label{sec:lowerBoundProof}
We begin with a lower bound of $\Omega(\Delta^2)$ on the number of colors for any one-round color reduction algorithm in the standard \LOCAL\ model (Section \ref{sec:OneRoundBound}), i.e., $\chi(\NH_1)\in \Omega(\Delta^2)$. This result was shown before in \cite{stacs09}, but our proof is much simpler and we believe that it is also instructive as it contains the core idea for the subsequent general lower bound proof for the \SETLOCAL\ model. 

Afterwards, the goal is to device a lower bound on the chromatic number of $\NCR_r$. For this purpose, we (recursively) define graphs $\NT_r$ and $\NB_r$. The recursive structure of the graph $\NT_r$ is simpler than the one of $\NCR_r$ such that the repetitive application of the ideas of Section \ref{sec:OneRoundBound} amplify to a lower bound on $\chi(\NT_r)$ (Section \ref{sec:MainLB}). 
Any graph homomorphism $h: G\rightarrow H$ implies \mbox{$\chi(G)\leq \chi(H)$} and in Section \ref{sec:homomorphisms}, we show that for the correct choice of parameters there is a chain of  homomorphisms $\NT_r\longrightarrow \NB_r\longrightarrow \NCR_r$.
Hence the lower bound on $\chi(\NT_r)$ translates into a lower bound on $\chi(\NCR_r)$.
In Section \ref{sec:runtimeLowerBound}  we combine all results to compute a runtime lower bound on any distributed color reduction algorithm in the \SETLOCAL\ model.  

\subsection{One-Round Lower Bound in the \boldmath\LOCAL\ Model}
\label{sec:OneRoundBound}
For a set $S$ let $A\sqsubseteq S$ denote that $A$ is a multiset consisting of elements of $S$.
For integers $\Delta\geq 2$ and $m>\Delta$, we define the one-round neighborhood 
graph $\NH_1(m,\Delta)$ with
\begin{eqnarray*}
  V\left(\NH_1(m,\Delta)\right) & := &
  \set{(x,A) | x\in[m], A\sqsubseteq [m], |A|\leq \Delta, x\not\in A}
\end{eqnarray*}
	and there is an edge between two nodes $(x,A), (y,B)\in \NH_1(m,\Delta)$ if $x\in B \text{ and }y\in A$.

The above definition is the most general version of one-round neighborhood graphs in the \LOCAL\ model; however, in \cite{Kuhn2006On} the authors show that for a single round it is sufficient to let $A$ be a simple subset of $[m]$ (not a multiset) having exactly $\Delta$ elements.
After all, we do not know whether multisets are necessary when extending the definition for more than a single round. 
\begin{thm}\label{thm:oneRound}
  For all $\Delta\geq 2$ and $m\geq \Delta^2/4 + \Delta/2 + 1$, we have
  $\chi\big(\NH_1(m,\Delta)\big) > \frac{\Delta^2}{4}$.
\end{thm}
The following proof captures the main idea of the constructions of Section \ref{sec:MainLB} in a simpler setting. In particular, it contains the main idea for the proof of Lemma \ref{lemma:cliqueConstruction}. Additionally, it provides an alternative characterization of the terms \emph{source} and \emph{non-source} (cf. Definition \ref{def:source}). We believe that this characterization gives a deeper understanding of subsequent proofs.

\begin{proof}
Let $I\subseteq V\big({\NH_1}\big)$ be an independent set of
  $\NH_1(m,\Delta)$. 
	Then $I$ induces an \emph{orientation $D_I$} of the
  edges of the complete graph $K_m$ on the vertices $[m]$ in the
  following way: For each $x,y\in [m]$, if there exists a node
  $(x,A)\in I$ for which $y\in A$, we say that the edge $\set{x,y}$ of
  $K_m$ is oriented from $x$ to $y$. 
	Because $I$ is an
  independent set, it is not possible that
  an edge $\set{i,j}$ is oriented in both directions. If $I$
  does not lead to an orientation of an edge $\set{x,y}$, we orient
  $\set{x,y}$ arbitrarily. We say that an independent set \emph{covers} a node $(x,A)$ if the edge $\{x,y\}$ is oriented towards $y$ for all $y\in A$. Clearly, $I$ covers all nodes with $(x,A)\in I$. 

	For a set $W\subseteq K_m$ we say that $x\in K_m$ is a \emph{$W$-source} of $I$
	 if for all $y\in W\setminus \set{x}$, the
  edge $\set{x,y}$ of $K_m$ is oriented from $x$ to $y$ (in
  $D_I$). If $W=K_m$ we simply call $x$ a \emph{source}.

  Now assume for contradiction that we are given a vertex coloring of
  $\NH_1(m,\Delta)$ that uses $c= \Delta^2/4$ colors. For
  each of the colors $k\in [c]$, the nodes $I_k$ colored with color
  $k$ form an independent set of $\NH_1(m,\Delta)$. The given
  $c$-coloring therefore also induces $c$ orientations $D_1,\dots,D_c$
  of $K_m$ such that for every $(x,A)\in V\left({\NH_1}\right)$, one of the
  $c$ orientations covers $(x,A)$.

  Let $S\subseteq [m]$ be the set of nodes of $K_m$, which are
  a source of some orientation $D_k, k\in[c]$. Note that every orientation has
  at most $1$ source and therefore $|S|\leq c$. For the remainder of
  the proof, we restrict our attention to integers in
  $\bar{S}=[m]\setminus S$.
  We first fix an arbitrary set $T\subseteq \bar{S}$ of size
  $|T|=\lfloor \Delta/2 \rfloor +1$. Because $m\geq
  \Delta^2/4 + \Delta/2 + 1$, such a set $T$ exists. Clearly, each orientation $D_k$ can have at most one
  $T$-source. By the pigeonhole principle there exists an $x\in T$ such that
  $x$ is a $T$-source for at most $c/|T|$
  orientations. W.l.o.g., assume that $x\in T$ is a $T$-source
  for orientations $D_1,\dots,D_q$, where $q\leq c/|T|$.

  We now construct a node $(x,A)\in V\big({\NH_1}\big)$ that is not
  covered by any of the $c$ orientations $D_1,\dots,D_c$. We start by
  adding all $\lfloor \Delta/2\rfloor$ elements of $T\setminus
  \set{x}$ to $A$. Because $x$ is a $T$-source only for
  orientations $D_1,\dots,D_q$, none of the remaining $c-q$
  orientations can cover $(x,A)$. We have to add additional elements
  to $A$ in order to make sure that the orientations $D_1,\dots,D_q$
  also do not cover $(x,A)$. As $T$ only consists of elements that are
  not sources of any of the orientations, for each orientation $D_k$, $k\in[c]$, there is an element $y_k\in[m]$ such that the edge
  $\set{y_k,x}$ of $K_m$ is oriented from $y_k$ to $x$. For each of
  the orientations $D_k\in\set{D_1,\dots,D_q}$, we pick such an
  element $y_k$ and add $y_k$ to $A$. In this way, we obtain a pair
  $(x,A)$ that is not covered by any of the orientations
  $D_1,\dots,D_q$. The size of $A$ is
  \[  |A| \leq |T|-1 + q \leq \left\lfloor \frac{\Delta}{2}\right\rfloor + 
  \frac{c}{|T|}  < \frac{\Delta}{2} + \frac{c}{\Delta/2} =
  \Delta  \]
  and thus, $(x,A)$ is a node of $\NH_1(m,\Delta)$ which is not covered by any independent set. In particular, it does not have a color, a contradiction.
\end{proof}

\subsection{Recursive Structure of the Neighborhood Graph}

\label{sec:MainLB}

In this section we study the recursive structure of $\NCR_r$ (the graph $\NCR_r$ can be built from $\NCR_{r-1}$). 
We define two recursively defined sequences of graphs, $(\NT_0,\NT_1,\ldots)$ and $(\NB_0,\NB_1,\ldots)$.  The graphs  $\NT_0$ and $\NB_0$ are equal to the $m$-node clique  on the nodes   $[m]$. The nodes of the remaining graphs of the sequences are built according to the following recursive procedure: For $i\geq 0$, in each of the two sequences, a node of the $(i+1)$-st graph  is created by a node $x$ of the $i$-th graph and a subset $A$ of its neighbors. The sequences differ in the way which combinations of $x$'s neighbors in the $i$-th graph are allowed to form the set $A$. 

To specify this we need to introduce some notation:
For $i\geq 0$  each node of the graph $\NT_{i+1}$ $\big(\text{or }\NB_{i+1}\big)$ will be of the form $(x,A)$,
 where $x\in \NT_{i}$ $\big(\text{or }\NB_{i}\big)$  and $A\subseteq \Gamma(x)$. Define the \emph{center} of a node $(x,A)$ as $z((x,A))=x$ and the \emph{types} of a node as $R((x,A))=A$. For any set of nodes $A$ let $z(A)=\{z(a) \mid a\in A\}$. 
For $x\in \NT_{0}~ \big(\text{or }\NB_{0}\big)$ we define $z(x)=\bot$ and $R(x)=\{\bot\}$.  
\begin{definitions}[]
\label{def:GraphDefinition}
Let $m, D$ and $r$ be fixed. 
$\NT_0(m,D):=\NB_0(m,D):=K_{m}$, i.e, the clique on the nodes $[m]$.
 For $0\leq i < r$ we have the following recursive definitions
\begin{eqnarray*}
V\big(\NT_{i+1}(m,D)\big)& := &\left\{(x,A) \mid x\in V(\NT_i(m,D)), A\subseteq \Gamma_{\NT_i}(x), |A|\leq D\right\},\\
V\big(\NB_{i+1}(m,D)\big) & := & \big\{(x,A) \mid x\in V(\NB_i(m,D)), A\subseteq \Gamma_{\NB_i}(x), |A|\leq D,  R(x)= z(A)\big\}.
\end{eqnarray*}
For $i\geq 1$ there is an edge between two nodes $(x,A), (y,B)\in \NT_i~(\text{or }\NB_i)$ if $x\in B \text{ and }y\in A$.
\end{definitions}
\begin{figure}
\centering
 \includegraphics[width=0.31\textwidth]{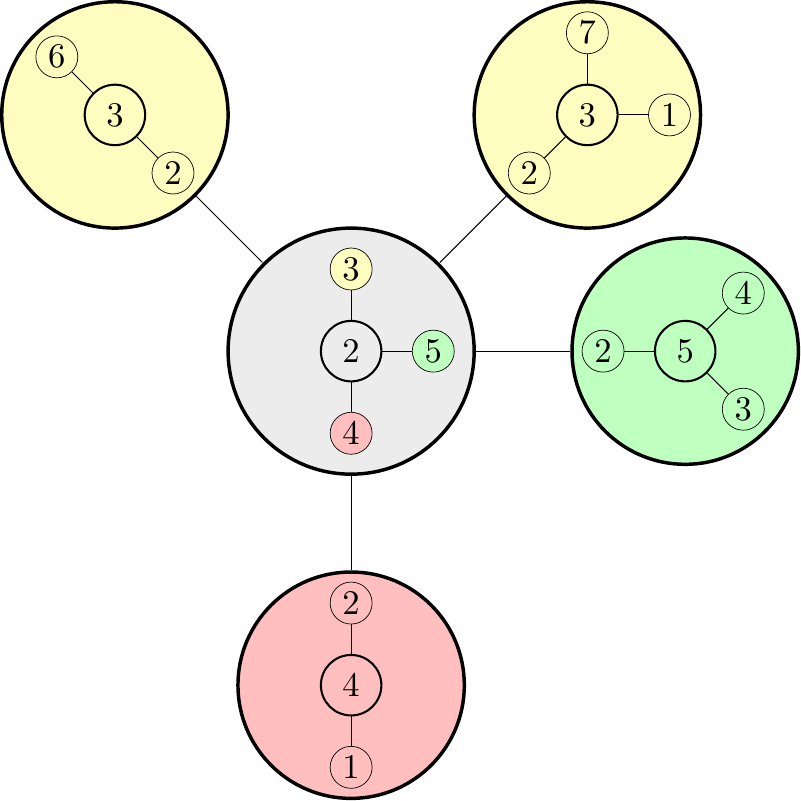}
\hspace{0.2cm}
\includegraphics[page=3,width=0.31\textwidth]{neighborhoodGraphs.pdf}
\hspace{0.2cm}
\includegraphics[page=2,width=0.31\textwidth]{neighborhoodGraphs.pdf}

\caption{Images (from left to right) explain the recursive structure of $\NB_2$, $\NH_2$ and $\NT_2$, respectively. The colors indicate types. The left image is the $2$-round view of the gray node in Figure \ref{fig:setModelTree} in the \SETLOCAL\ model. Here, after a single round node $2$ only knows that its degree is at least $3$, after the second round it learned that its degree is at least $4$ as it can now distinguish two neighbors with color $3$.
In the second image there is a neighbor of every type of the center with the corresponding multiplicity. For feasible nodes of $\NT_2$ the combination of neighbors is arbitrary w.r.t. to types of the center, e.g., in the third image there is no node for type $4$, i.e., no red node. We are not aware of a computational model to motivate $\NT_2$.} 
\label{fig:neighborhoodGraphs2}
\end{figure}
We denote $\NT_{i+1}(m,D)$ and $\NB_{i+1}(m,D)$ simply as  $\NT_{i+1}$ and $\NB_{i+1}$ whenever $m$ and $D$ are clear from the context. There is no restriction on the set $A\subseteq \Gamma_{\NT_i(m,D)}(x)$ to build a node of $\NT_{i+1}(m,D)$ as long as the size of $A$ is at most $D$. For a node \mbox{$(x,A)\in \NB_{i+1}(m,D)$} the set $A$ needs at least one fitting element for every type of $x$ (cf.~Figure \ref{fig:neighborhoodGraphs2}). When defining neighborhood graphs for the standard \LOCAL\ model (which we do not do in this paper for $r>1$) this restriction is even tighter. Then $A$, $R(x)$ and $z(A)$ need to be multisets and $A$ needs exactly one fitting element for every type of $x$.  
We are not aware of a  computational model which motivates the sequence $(\NT_0,\NT_1,\NT_2,\ldots)$.

In the proof of the lower bound on $\chi(\NT_r)$ we assume that we have a proper $c$-coloring of $\NT_r$ which implies (partial) $c$-colorings of the graphs $\NT_i, i<r$, cf. \mbox{Lemma \ref{lemma:partialColorings} and \mbox{Corollary \ref{corr:inducedColorings}}}. If $c$ is too small, we can use these partial colorings to construct an uncolored node $(x,A)\in \NT_r$, i.e., a contradiction. We begin with constructing an uncolored clique of $\NT_0$, which implies a smaller uncolored clique in $\NT_1$ and then a smaller uncolored clique in $\NT_2$ and so on until we reach an uncolored node in $\NT_r$.

The construction of a single uncolored node of $\NT_{i+1}$, denoted as $(x,A)$, is similar to the proof in \mbox{Section \ref{sec:OneRoundBound}}: For its construction we pick a suitable center node $x$ from the uncolored clique of $\NT_i$ and then (carefully) select up to $D$ neighbors to form the set $A$ which will ensure that the resulting node is uncolored. To iterate the argument we need to construct an uncolored clique of $\NT_{i+1}$ instead of a single uncolored node.

A center $x$ is suitable for the above process (this corresponds to a non-source in Section \ref{sec:OneRoundBound}) if for every color there exists a neighbor which, if contained in $A$,  implies that $(x,A)$ does not have this color. The (partial) $c$-coloring of $\NT_{i}$ induced by the (partial) $c$-coloring of $\NT_{i+1}$ is formed by all unsuitable centers. The following definitions/lemmata make this more precise.

\begin{definitions}[$W$-source, source]
\label{def:source}
For any set $I\subseteq \NT_{i+1}$ and any set $W\subseteq \NT_i$, a node 
 $x\in \NT_{i}$ is called a \textbf{\boldmath$W$-source} of $I$ if 
\[\forall w\in W\cap \Gamma_{\NT_i}(x): \exists (x,A)\in I \text{ with } w\in A.\]
If $W=\NT_i$ we call $x$ simply a \textbf{source} of $I$.
\end{definitions}

In this section we define \emph{sources} without the orientations from Section \ref{sec:OneRoundBound} because the proofs of consecutive statements become shorter, the characterization in Definition \ref{def:source} is slighty more general and possibly, it might even be further generalized in order to obtain a lower bound in the standard \LOCAL\ model.

Another intuition for the definition of sources is the following: Given a $c$-coloring of $\NT_{i+1}$ one realizes that $c$ can only be small if many nodes which have the same center have the same color. Note that nodes with the same center can never be adjacent. A natural approach would be to take a closer look at all centers $x$ which already determine the color of any node $(x,A)$. However, this is too restrictive and the centers which are sources generalize this approach.

In the following we show how a $c$-coloring of $\NT_r$ implies partial $c$-colorings of $\NT_i, i<r$.
For a set $I\subseteq \NT_{r}$ define $S_r(I):=I$ and for $i=r-1,\ldots,0$ inductively define  the following sequence of sets.
\begin{align*}
S_i(I):=\{x\in \NT_i \mid x \text{ is source of $S_{i+1}(I)$}\}.
\end{align*}
\begin{lemma}
\label{lemma:partialColorings}
Let $I$ be an independent set of $\NT_r$. 

\noindent Then for all $i=0,\ldots,r$ the set 
$S_i(I)$ is an independent set of $\NT_i$. 
\end{lemma}
\begin{proof}
We prove the result by (backwards) induction on $i=r,\ldots,0$: the statement holds trivially for $i=r$. 
For the induction step let the statement be true for $i+1$.
Now, assume that $S_i(I)$ is not an independent set, i.e., there are two nodes $x,y\in S_i(I)$ with $\{x,y\}\in E(\NT_i)$.
As $x$ and $y$ are sources of $S_{i+1}(I)$ there exist nodes $(x,A_x),(y,A_y)\in S_{i+1}(I)$ with $y\in A_x$ and $x\in A_y$. Hence $\{(x,A_x),(y,A_y)\}\in E(\NT_{i+1})$, which is a contradiction to $S_{i+1}(I)$ being an independent set.
\end{proof}
Each color class of a $c$-coloring is an independent set. Thus any (partial) $c$-coloring corresponds to $c$ (disjoint) independent sets.
Vice versa any collection of $c$ independent sets induces a partial $c$-coloring though a single node might have more than one color. It is still a coloring in the sense that any of those colors is different from any color of its neighbors.  With this identification of colorings and independent sets we obtain the following corollary.

\begin{corr}
\label{corr:inducedColorings}
Let $I_1,\ldots, I_c$ be a $c$-coloring of $\NT_r$.
 
\noindent Then $S_i(I_1), \ldots, S_i(I_c)$ corresponds to a (partial) $c$-coloring of $\NT_i$.
\end{corr}
The following properties of sources are needed for the construction of uncolored nodes.

\begin{lemma}
\label{lemma:simpleSource}
Let $W\subseteq \NT_{i}$ be a clique and $I_1,\ldots, I_c$ independent sets of $\NT_{i+1}$. For all $k\in [c]$ we have
\begin{myenumerate}[(a)]
\item If $x\in \NT_{i}$ is not a source of $I_k$ there exists $w\in \Gamma_{\NT_i}(x)$ such that 
\[\text{$(x,A)\notin I_k$, whenever $w\in A$.}\]
\item There is at most one $W$-source of $I_k$ within $W$.
\item There exists a node $x\in W$ which is a $W$-source for at most $\frac{c}{|W|}$ of the independent sets $I_1,\ldots, I_c$. 
Furthermore, for any choice of $A$ with $W\setminus\{x\}\subseteq A$ we have $(x,A)\notin I_k$ for all but $\frac{c}{|W|}$ many independent sets.
\end{myenumerate}
\end{lemma}

\begin{proof}
\emph{Proof of $a)$:} This follows by negating the definition of a source. 
\emph{Proof of $b)$:} $b)$ holds because $I_k$ is an independent set and due to Lemma \ref{lemma:partialColorings} the sources form an independent set, i.e., there can be at most one in the clique $W$. This does not change change if we consider $W$-sources. 
\emph{Proof of $c)$:} By $b)$ there is at most one $W$-source in $W$ for each of the $c$ independent sets. Thus there is a node $x\in W$ which is a $W$-source for at most $\frac{c}{|W|}$ independent sets by the pigeonhole principle. The second statement follows because whenever there is a node $(x,A)\in I_k$ and $W\setminus\{x\}\subseteq A$ then, by definition, $x$ is a $W$-source of $I_k$.
\end{proof}
For sets $I_1,\ldots,I_c\subseteq\NT_r$ we call a subset $T\subseteq \NT_i$ \emph{uncolored} if we have $T\cap S_i(I_k)=\emptyset$ for all $k\in [c]$.

\begin{lemma}
\label{lemma:cliqueConstruction}
Let $p, c$ and $d$ be positive integers and let $I_1,\ldots,I_c$ be independent sets of $\NT_r$ with
\begin{align}
\label{cond:CliqueCondition}
p+d-1+\frac{c}{d}\leq D.
\end{align}
Any uncolored  clique $T\subseteq \NT_i$ of size $p+d$  leads to an uncolored  clique $T'\subseteq \NT_{i+1}$ of size $p$.
\end{lemma}
\begin{proof}

We inductively determine nodes $t_1,\ldots, t_{p}\in T$ which will form the centers of the clique nodes in $\NT_{i+1}$. Assume that nodes \mbox{$t_1,\ldots,t_{j-1}\in T$} are already determined and let $T_j$ be any subset of $T\setminus \{t_1,\ldots,t_{j-1}\}$ with size $d$. Such a set exists because $|T\setminus\{t_1,\ldots,t_{j-1}\}|\geq p+d-(j-1)\geq d$.

$T_j$ is a clique in $\NT_i$ and by Lemma \ref{lemma:partialColorings} the sets $S_{i+1}(I_1),\ldots, S_{i+1}(I_c)$ are independent sets in $\NT_{i+1}$. Hence there exists a node in $T_j$ which is a $T_j$-source for at most $q\leq\frac{c}{|T_j|}=\frac{c}{d}$ independent sets by \mbox{Lemma \ref{lemma:simpleSource} $(c)$}. Denote this node by $t_j$ and continue with determining the node $t_{j+1}$.

After determining $t_1,\ldots, t_{p}$ construct the uncolored clique of $\NT_{i+1}$ as follows: For \mbox{$j=1,\ldots,p$} let 
\begin{align*}
x_j:=(t_j,A_j)=\left(t_j,\left(T\setminus\{t_j\}\right)\cup B_j\right), 
\end{align*}
where $B_j$ will be constructed later. Regardless of the choice of the $B_i$'s the nodes $x_1,\ldots,x_{p}$ form a clique because $t_{j'}\in A_j$ and $t_j\in A_{j'}$ for $j\neq j'$.

We argue how to choose the set $B_j$ such that $x_j$ is uncolored in $\NT_{i+1}$. 
Due to the choice of $t_j$ and $T_j\setminus\{t_j\}\subseteq A_j$ all but $q$ independent sets do not contain $x_j$ and we eliminate each of those one-by-one with the choice of $B_j$.
W.l.o.g. let the remaining independent sets be $S_{i+1}(I_1),\ldots,S_{i+1}(I_{q})$.  Because $t_j\in T$ is uncolored in $\NT_{i}$ it is not a source for any of the independent sets $S_{i+1}(I_k), k\in [q]$. Hence via \mbox{Lemma \ref{lemma:simpleSource} $(a)$} there exist  $b_1,\ldots,b_{q}\in \Gamma_{\NT_i}(t_j)$ such that $(t_j,A_j)$ is not contained in any of the independent sets $S_{i+1}(I_k), k\in[q]$, whenever \mbox{$B_j:=\{b_1,\ldots,b_{q}\}\subseteq A_j$}. Hence $x_j$ is uncolored.

The node $x_j$ actually is a valid node of $\NT_{i+1}$ (cf. Definition \ref{def:GraphDefinition}), as  $A_j\subseteq \Gamma_{\NT_i}(t_j)$ and 
\begin{align*}
|A_j|=|T\backslash \{t_j\} \cup B_j| \leq p+d-1+\frac{c}{d}\stackrel{(\ref{cond:CliqueCondition})}{\leq} D.\label{conditionToFormANode}
\end{align*}
Hence $x_1,\ldots, x_{p}$ is an uncolored clique of $\NT_{i+1}$.
\end{proof}
An identical proof for the neighborhood graphs in the \LOCAL\ model fails in the last step because the newly constructed node might not be a node of $\NH_{i+1}$ due to mismatching types (cf. the comment after \mbox{Definition \ref{def:GraphDefinition}}). 
\begin{thms}
\label{theorem:lowerBoundNT}
Let $m\geq \frac{D^2}{4r}+\frac{D}{2}+1$. Then the following lower bound on the chromatic number holds
\begin{align}
\chi(\NT_{r}(m,D))> \frac{D^2}{4r}.
\end{align}
\end{thms}
\begin{proof}\footnote{Note that parameters here and in Lemma \ref{lemma:cliqueConstruction} are not optimized to the very last because we are only interested in asymptotic behaviour.}
Assume that there is a coloring of $\NT_r(m,D)$ with $c=\frac{D^2}{4r}$ colors, i.e., there are
independent sets $I_1,\ldots, I_c$  which partition the node set of $\NT_r$.
For $i=0,\ldots,r$ the  sets 
$S_i(I_1),\ldots, S_i(I_c)$ form independent sets and a partial coloring with $c$ colors of $\NT_i$ by Lemma \ref{lemma:partialColorings} and Corollary \ref{corr:inducedColorings}.

Let $d=\frac{D}{2r}$. With the help of Lemma \ref{lemma:cliqueConstruction} we inductively (induction on $i=0,\ldots,r$) construct uncolored cliques of size $rd-id+1$ in $\NT_i$ w.r.t. the respective partial coloring induced by $S_i(I_1),\ldots, S_i(I_c)$.\\
\textbf{Base case:}
By Lemma \ref{lemma:simpleSource} $(b)$ and because $\NT_0$ is a clique, the partial $c$-coloring $S_0(I_1),\ldots ,S_0(I_c)$ of $\NT_0$ can color at most $c$ nodes. As $m\geq c+rd+1$ there are $rd+1$ nodes in $\NT_0$ which form an uncolored clique. \\
\textbf{Induction step:} Let $p=rd-(i+1)d+1$ and let $T$ be an uncolored clique of size $rd-id+1=p+d$  in $\NT_i$.
Then Condition (\ref{cond:CliqueCondition}) in \mbox{Lemma \ref{lemma:cliqueConstruction}}  is satisfied because
\[p+d-1+\frac{c}{d}=rd-id+\frac{c}{d}\leq rd+\frac{c}{d}=\frac{D}{2}+\frac{D}{2}\leq D.\]
With \mbox{Lemma \ref{lemma:cliqueConstruction}} we obtain an uncolored clique of size $p$ in $\NT_{i+1}$. By the principle of induction the result holds for all $i$ and there  is an uncolored clique of size $1$ in $\NT_r$, i.e., an uncolored node, a contradiction.
\end{proof}

\subsection{Graph Homomorphisms}
\label{sec:homomorphisms}
The existence of a graph homomorphism from $\NB_r(m,D)$ to $\NCR_r(m,D)$ is intuitive as the recursive structure of both graphs is exactly the same.
\begin{lemma}\label{lemma:graphHomo}
 There is a graph homomorphism
\begin{align*}
h_r: \NB_r(m,D)\rightarrow \NCR_r(m,D).
\end{align*}
\end{lemma}
\begin{proof}

Formally the center function $z$ is a function with domain $\NT_{j+1}$ and range $\NT_{j}$ for some fixed $j$; so it should be indexed as $z_j$. However, in the course of this proof we slightly missuse notation and for $x\in \NT_r$ we define $z^{r-i}(x):=z_{i}(z_{i+1}(\ldots z_{r-2}(z_{r-1}(x))\ldots))\in \NT_{i}$ as the element which one obtains when recursively applying $z_j$ with $j=r-1,\ldots,i$.

We perform a constructive proof with an induction on $r$. 

\noindent\textbf{Base case:}
$\NB_0(m,D)$ and $\NCR_0(m,D)$ are both the complete graph on the nodes $[m]$. Define
\begin{align*}
h_0(i):=i\text{ for }i\in [m].
\end{align*}
$h_0$ is well defined and a graph homomorphism.

\noindent\textbf{Induction step:}
Let $(x,A)\in \NB_{r+1}(m,D)$, that is $x\in \NB_r(m,D)$ and $A\subseteq \Gamma_{\NB_r}(x)$ and define 
\begin{align*}
h_{r+1}((x,A)):=\left(h_r(x), \left\{h_r(a) \mid a\in A\right\}\right).
\end{align*}
\begin{figure}[htp]
\centering
\includegraphics[width=0.8\textwidth]{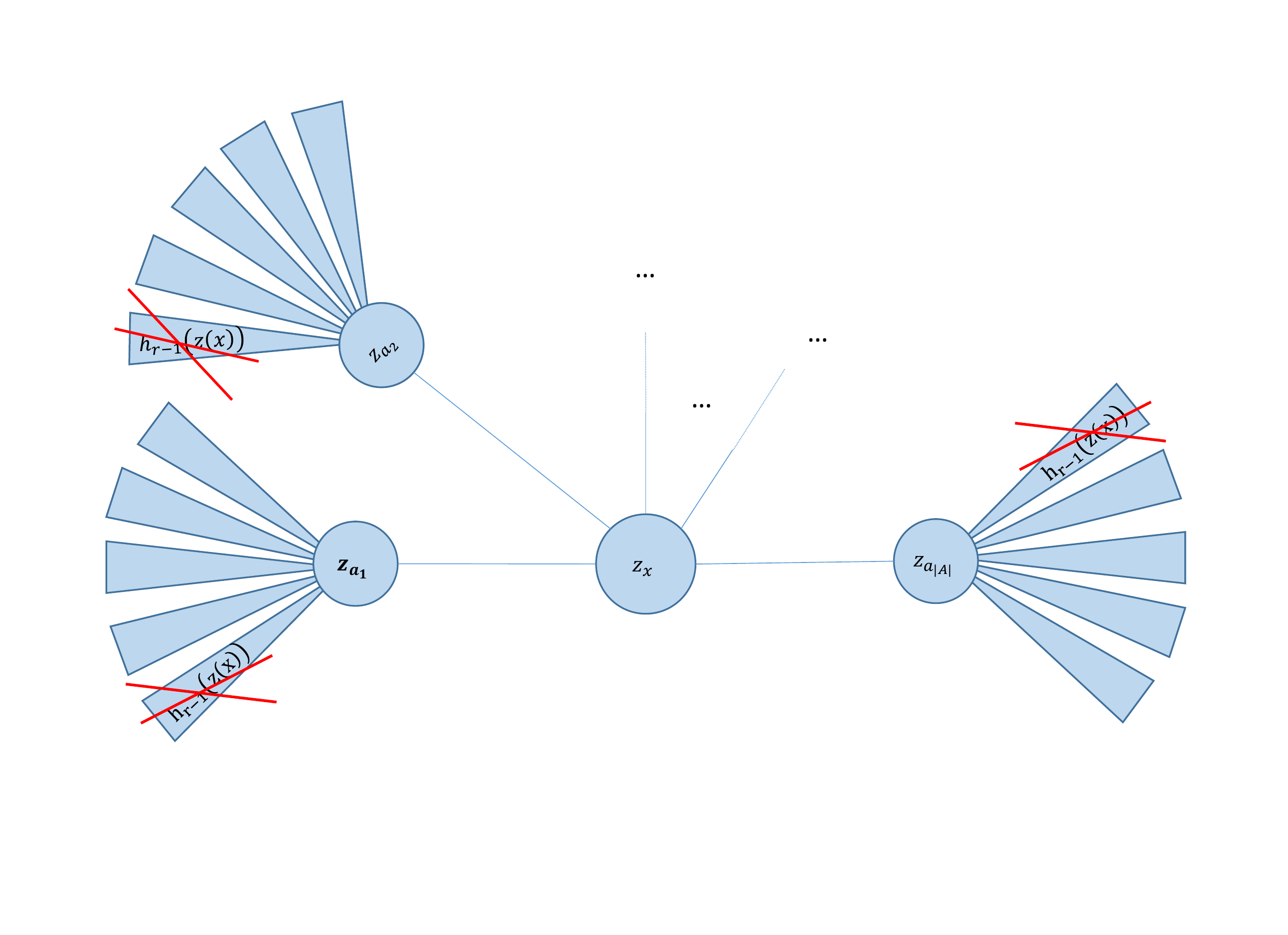}
\label{fig:lemma41Proof}
\caption{In the proof of Lemma \ref{lemma:graphHomo} there is a tree $H_a$ for every $a\in A$ and a tree $H_x$ for $x$. We cut off one branch with $(r-1)$-view $h_{r-1}(z(x))$. Then we connect all trees to the node $z_x$ and show that $z_x$ and all $z_a, a\in A$ still have the same $r$-views as in the respective trees.}
\end{figure}
\noindent\textbf{Claim 1:} $h_{r+1}(x,A)\in \NCR_{r+1}(m,D)$.
\begin{proof}[Proof of Claim 1]
We need to show that $h_{r+1}((x,A))$ actually is a node of $\NCR_{r+1}$, that is, there exists a tree $G_{(x,A)}$ with maximum degree $D$ and a node which has $(r+1)$-view $h_{r+1}((x,A))$.

Construct the following tree:  
Define $z_x:=z^{r}(x)\in \NB_0$ and for $a\in A$ define \mbox{$z_a:=z^{r}(a)\in \NB_0$}.
By the induction hypothesis there are trees $H_a$ with root $z_a$ for every $a\in A$ such that the $r$-view of $z_a$ in $H_a$ is $h_r(a)$. 
In $H_a$ the root $z_a$ has at least one neighbor with $(r-1)$-view $h_{r-1}(z(x))$. 
 Let $\TG_a$ denote the tree $H_a$ where the subtree starting at that node is removed.
Build a new graph $G_{(x,A)}$ with node and edge set
\begin{eqnarray*}
V(G_{(x,A)})& := & \{z_x\}\dcup \bigcup_{a\in A} V(\TG_a),\\
E(G_{(x,A)}) & := & \bigcup_{a\in A} \{z_x, z_a\}\dcup\bigcup_{a\in A} E(\TG_a).
\end{eqnarray*}
For any $a\in A$ the degree of $z_a$ in $\TG_a$ is at most $D-1$. Hence the degree of $z_a$ in $G_{(x,A)}$ is at most $D$. Because $|A|\leq D$ the degree of $z_x$ is at most $D$ and all remaining nodes have degree at most $D$ as well.

 We need to show that the $(r+1)$-view of $z_x$ in $G_{(x,A)}$ equals $h_{r+1}(x,A)$. Because there is a node $a\in A$ for every type of $x$ (, i.e., $z(A)= R(x)$, cf. Definition \ref{def:GraphDefinition})  it is sufficient to show that the $r$-view of $z_x$ is $h_r(x)$ and the $r$-view of $z_a$ in $G_{(x,A)}$ is $h_r(a)$ for all $a\in A$, that is, we need to show
\begin{align}
\label{cond:WhatToShow}
S^{G_{(x,A)}}_{r+1}(z_x)=\left(S^{G_{(x,A)}}_{i}(z_x), \left\{S^{G_{(x,A)}}_{i}(z_a) \mid a\in A\right\}\right)\stackrel{!}{=} \left(h_{r}(x), \left\{h_r(a)\mid a\in A\right\}\right)=h_{r+1}(x,A).
\end{align}
At first observe that if the $i$-views of two nodes are the same also their $j$-views are the same for $j\leq i$.
To prove (\ref{cond:WhatToShow}) we use a strong induction on $0\leq i < r$ with induction hypothesis  
\begin{align}
\label{cond:strongInduction}
S^{G_{(x,A)}}_i(z_x)& =h_i\left(z^{r-i}(x)\right) &\text{ and } &  & S^{G_{(x,A)}}_i(z_a)& =h_i\left(z^{r-i}(a)\right)\text{ for }a\in A.
\end{align}
 \textbf{Base case (Claim 1):} The result holds trivially for $i=0$, because $h_0(z^{r-0}(x))=z_x$, which equals the $0$-view of $z_x$ in the tree $G_{(x,A)}$, and $h_0(z^{r-0}(a))=z_a$, which equals the $0$-view of $z_x$ in the tree $G_{(x,A)}$. \\
\textbf{Induction step (Claim 1):} Assume that (\ref{cond:strongInduction}) holds for some $i$; we show that it holds for $i+1$ as well. 
As the $i$-views of all $z_a, a\in A$ are equal to $h_i(z^{r-i}(a))$ and the $i$-view of $z_x$ equals to $h_i\left(z^{r-i}(x)\right)$ by the induction hypothesis  we can deduce 
\begin{align*}
S^{G_{(x,A)}}_{i+1}(z_x) & =\left(S^{G_{(x,A)}}_{i}(z_x), \left\{S^{G_{(x,A)}}_{i}(z_a) \mid a\in A\right\}\right)\\
& = \left(h_{i}(z^{r-i}(x)), \left\{h_i(z^{r-i}(a))\mid a\in A\right\}\right)\\
& = h_{i+1}(z^{r-i-1}(x)).
\end{align*}
Furthermore, by definition  we have
\begin{align}
\label{line1}
S^{G_{(x,A)}}_{i+1}(z_a) & =\left(S^{G_{(x,A)}}_{i}(z_a), \left\{S^{G_{(x,A)}}_{i}(u) \mid u\in \Gamma(z_a)\right\}\right).
\intertext{We show that the $i$-view of $u\in \Gamma(z_a)$ is the same in $G_{(x,A)}$ and $H_a$: If $u\neq z_x$ this follows because the $(i-1)$-view of $z_a\in \Gamma(u)$ is the same in $H_a$ and $G_{(x,A)}$ due to the induction hypothesis. For $u=z_x$ the $(i-1)$-view  equals to the $(i-1)$-view of the deleted subtree by the induction hypothesis. Hence (\ref{line1}) equals}
& =\left(S^{G_{(x,A)}}_{i}(z_a), \left\{S^{H_a}_{i}(u) \mid u\in \Gamma(z_a)\right\}\right). \label{line2}
\intertext{The $i$-view of $z_a$ is the same in both graphs because its $(i-1)$-view is the same, the $(i-1)$-view of $z_x$ is the same and the $(i-1)$-view of every neighbor $u\neq z_x$ is the same in both graphs. Thus (\ref{line2}) equals}
& =\left(S^{H_a}_{i}(z_a), \left\{S^{H_a}_{i}(u) \mid u\in \Gamma(z_a)\right\}\right)\nonumber\\
& =S^{H_a}_{i+1}(z_a)=h_{i+1}(z^{r-i-1}(a)).\nonumber
\end{align}
\end{proof}
\noindent\textbf{Claim 2:} Every edge in $\NB_{r+1}$ is mapped to an edge in $\NCR$.
\begin{proof}[Proof of Claim 2]
Let $(x,A), (y,B)\in \NB_{r+1}$ with $x\in B$ and $y\in A$. Define $z_x:=z^{r}(x)\in \NB_0$ and $z_y:=z^{r}(y)\in \NB_0$. 
We need to show that there exists a graph with two neighboring nodes which have the views $h_{r+1}((x,A))$ and $h_{r+1}((y,B))$.

We combine the graphs $G_{(x,A)}$ and $G_{(y,B)}$ to obtain a graph $G$ in which their roots $z_x$ and $z_y$ are neighbors and views are corresponding. Let $\TG_{(x,A)}$ be the graph $G_{(x,A)}$ without the subgraph starting at a child of $z_x$ which has $r-1$-view $h_{r-1}(y)$ and let  $\TG_{(y,B)}$ be the graph $G_{(y,B)}$ without the subgraph starting at a child of $z_y$ which has $r-1$-view $h_{r-1}(x)$.
Then consider the following graph $G=(V(G), E(G))$:
\begin{eqnarray*}
V(G)& := &\{z_x, z_y\}\dcup V\left(\TG_{(x,A)}\right)\dcup V\left(\TG_{(y,B)}\right)\\
E(G) & := &\{\{z_x,z_y\}\}\dcup E\left(\TG_{(x,A)}\right)\dcup E\left(\TG_{(y,B)}\right).
\end{eqnarray*}

A proof which shows that the views of $z_x$ and $z_y$ are $h_{r+1}((x,A))$ and $h_{r+1}((y,B))$, respectively, is along similar lines as the proof of Claim 1.
\end{proof}
This concludes the proof of both claims and thus the induction step and the proof of Lemma \ref{lemma:graphHomo}.
\end{proof}

\begin{lemma}
\label{lemma:graphHomo2}
For $r\geq 0$ there is a  graph homomorphism 
\begin{align*}
f_r: \NT_r(m,D)\rightarrow \NB_r\left(m,(r+1)D\right)
\end{align*}
\end{lemma}

\begin{proof}[Proof sketch]We show the result by induction on $r$. 

\noindent\textbf{Base case:} 
For $r=0$ the graphs $\NT_0(m,D)$ and $\NB_0(m,D)$ are both equal to the clique on the nodes $[m$] which implies the existence of $f_0$.

\noindent\textbf{Induction step:}
Assume that $f_r: \NT_r(m,D)\rightarrow \NB_r\left(m,(r+1)D\right)$ is given. We construct
\begin{align*}
f_{r+1}: \NT_{r+1}(m,D)\rightarrow \NB_{r+1}\left(m,(r+2)D\right).
\end{align*}
Let $(x,A)$ and $(y,B)$ be neighboring nodes in $\NT_{r+1}(m,D)$, i.e., $x\in B$ and $y\in A$. 
\begin{align*}
f_{r+1}((x,A)) & :=\left(f_r(x), A'\cup \bigcup_{a\in A}f_r(a)\right), &
f_{r+1}((y,B)) :=\left(f_r(y), B'\cup \bigcup_{b\in B}f_r(b)\right),
\intertext{
where  $A',B'\subseteq \Gamma_{\NB_r(m,(r+1)D)}(f_r(x))$ are such that they fill up the types of $f_r(x)$ (or $f_r(y)$) which are not already met by $\left(\bigcup_{a\in A}f_r(a)\right)$ (or $\left(\bigcup_{b\in B}f_r(b)\right)$).
More precisely $A'$ and $B'$ are such that}
z(A') & =R(f_r(x))\setminus \left(\bigcup_{a\in A}f_r(a)\right), & \left|A'\cup \left(\bigcup_{a\in A}f_r(a)\right)\right|\leq (r+2)D,  \\
z(B') & =R(f_r(y))\setminus \left(\bigcup_{b\in B}f_r(b)\right), & \left|B'\cup \left(\bigcup_{b\in B}f_r(b)\right)\right|\leq (r+2)D.
 \end{align*}
$A'$ and $B'$ can be chosen arbitrarily as long as they fulfill the above requirements. Note that the type requirements (the ones on the left hand side) can be fulfilled easily. To make the choice of $A'$ and $B'$ unique simply assume any linear order on $\NB_r(m,(r+1)D)$ and let $A'$ and $B'$ be the minimal sets which satisfy the above conditions.
Further note that the size restriction on the right hand side can be met because 
\[\left|A'\cup \left(\bigcup_{a\in A}f_r(a)\right)\right|\leq \left|R(f_r(x))\right|+|A|\leq (r+1)D+D\leq (r+2)D.\]
The reasoning is identical for $B'$.

Note that $f_{r+1}((x,A))$ and $f_{r+1}((y,B))$ are neighbors in $\NB_{r+1}(m,(r+2)D)$.
\end{proof}
The \SETLOCAL\ model is crucial in the proof of the homomorphism in Lemma \ref{lemma:graphHomo2}. The recursive construction of homomorphisms breaks when considering $\NH_r$ instead of $\NB_r$.
\subsection{Proof of Theorem \ref{thm:lowerbound}}
\label{sec:runtimeLowerBound}
Any graph homomorphism $f:G\rightarrow H$ implies that $\chi(G)\leq \chi(H)$ and we devised graph homomorphisms
\begin{align*}
\NT_r(m,D)\stackrel{f_r}{\longrightarrow} \NB_r(m,(r+1)D) \stackrel{h_r}{\longrightarrow} \NCR_r(m,(r+1)D) \stackrel{h}{\longrightarrow} \NCR_r(m,2rD).
\end{align*}
The existence of the last homomorphism is trivial and together with  Theorem \ref{theorem:lowerBoundNT} this implies
\begin{align}\label{cond:chromaticNumber}
\chi\left(\NCR_r(m,2rD)\right)> \frac{D^2}{4r}, ~\text{for } m\geq \frac{D^2}{4r}+\frac{D}{2}+1.
\end{align}
To prove Theorem \ref{thm:lowerbound} assume an $r$-round $(C \Delta^{1+\eta})$-coloring algorithm.
Set the parameter $D:=\big(2C\Delta^{2+\eta}\big)^{\frac{1}{3}}$. Then the condition on $m$ in (\ref{cond:chromaticNumber}) is satisfied for $\Delta$ sufficiently large. With $r:=\big(\frac{1}{16C}\Delta^{1-\eta}\big)^{\frac{1}{3}}$
statement (\ref{cond:chromaticNumber}) implies the contradiction
$\chi\big(\NCR_{r}(m,\Delta)\big)> C\Delta^{1+\eta}$.\hspace*{\fill}\qed

\subsection{Lower Bound for One-Round Defective-Coloring Algorithms}
\label{sec:LowerBoundWeakColoring}
In this section we show an $\Omega\big(\big(\frac{\Delta}{d+1}\big)^2\big)$ lower bound on the number of colors of a $d$-defective one-round coloring algorithm in the \LOCAL\ model starting with an initial $m$-coloring for sufficiently large $m$.
A $d$-defective $c$-coloring of a graph $G=(V,E)$ is a function $\phi: V\rightarrow [c]$ such that each color class induces a graph with maximum degree $d$, i.e., let $S_i=\{v\in V \mid \phi(v)=i\}$ then the induced graph $G[S_i]$ has maximum degree $d$. The $d$-defective chromatic number $\chi^d(G)$ of a graph $G$ is the smallest integer $c$ such that there is a $d$-defective $c$-coloring of $G$. Note that a $0$-defective $c$-coloring is a proper $c$-coloring in the usual sense.

Much of this section is similar to Section \ref{sec:lowerBoundProof} and in particular Section \ref{sec:OneRoundBound}. One can unify all definitions for those parts. However, to keep the proof of the main result in Section \ref{sec:lowerBoundProof} as simple as possible and to make Section \ref{sec:LowerBoundWeakColoring} self contained we repeat definitions in this chapter and slighlty adapt them to be suitable for defective colorings.

For integers $\Delta\geq 2$ and $m>\Delta$, we define the  
graph $\NH_1(m,\Delta)$ with
\begin{eqnarray*}
  V\left(\NH_1(m,\Delta)\right) & := &
  \set{(x,A) | x\in[m], A\sqsubseteq [m], |A|\leq \Delta, x\not\in A}
\end{eqnarray*}
	and there is an edge between two nodes $(x,A), (y,B)\in \NH_1(m,\Delta)$ if $x\in B \text{ and }y\in A$.
 
Similar to Lemma \ref{lemma:algCorrespondence} we obtain the following lemma.
\begin{lemma}\label{lemma:algDefectiveCorrespondence}
Any deterministic one-round distributed algorithm in the \LOCAL\ model, which correctly $d$-defectively $c$-colors any initially $m$-colored graph with maximum degree $\Delta$, yields a feasible $d$-defective $c$-coloring of $\NH_1(m,\Delta)$ and vice versa.
\end{lemma}
\begin{proof}
Assume that we are given a one-round distributed algorithm in the \LOCAL\ model, which $d$-defectively $c$-colors any initially $m$-colored graph with maximum degree $\Delta$. This algorithm induces a $c$-coloring $\phi$ of the nodes of $\NH_1(m,\Delta)$. Assume that this $c$-coloring has a defect larger than $d$, i.e., there are $d+2$ nodes $x_0,\ldots,x_{d+1}\in V(\NH_1(m,\Delta))$ with $\phi(x_0)=\ldots=\phi(x_{d+1})$ and $x_1,\ldots,x_{d+1}\in \Gamma(x_0)$. By the definition of $\NH_1(m,\Delta)$ we can construct an initially $m$-colored tree $G$ with maximum degree $\Delta$ where the one-round view $x_0$ occurs as the view of a node $v_0\in G$ and the one-round views $x_1,\ldots,x_{d+1}$ occur as one-round views of nodes $v_1,\ldots,v_{d+1}\in \Gamma_G(v)$.  As a consequence the algorithm corresponding to the function $\phi$ assigns the same value to $v_0,\ldots,v_{d+1}$, which is a contradiction to its correctness.

For the other direction we need to prove that any $r$-round $c$-coloring of $\NH_1(m,\Delta)$ implies a one-round $c$-coloring algorithm of any initially $m$-colored graph $G$ with maximum degree $\Delta$. Assume that we are given a network graph $G$ with initial $m$-coloring $\psi$. For that purpose assume that we have a $c$-coloring $\phi$ of $\NH_1(m,\Delta)$ which is known by all nodes of $G$. This is no problem because the graph $\NH_1(m,\Delta)$ does only depend on $m$ and $\Delta$ and is independent of the structure of any particular graph $G$. In one round let each node $v$ collect the set of its neighbors' colors $\Gamma_v=\{\psi(u) \mid \{v,u\}\in E(G)\}$. Then $(\psi(v),\Gamma_v)$ is a node of $\NH_1(m,\Delta)$. Let $v$ select color $\phi((\psi(v),\Gamma_v))$. If a node $v_0\in G$ has $d+1$ distinct neighbors $v_1,\ldots, v_{d+1}\in G$ their selected colors cannot be identical because the nodes  $(\psi(v_1),\Gamma_{v_1}), \ldots, (\psi(v_{d+1}),\Gamma_{v_{d+1}})$ are neighbors of $(\psi(v_0),\Gamma_{v_0})$ in  $\NH_1(m,\Delta)$.
\end{proof}
There are several problems with extending  Lemma \ref{lemma:algDefectiveCorrespondence} for more than one round.
For two or more rounds a $d$-defective coloring algorithm in the \LOCAL\ model does not (automatically) correspond to a $d$-defective coloring of the corresponding neighborhood graph: Neighbors of a node $x$ in $\NH_r$ with the same type can have the same color as $x$  and still only contribute for a single defect of the algorithm. 
In the \SETLOCAL\ model a $d$-defective coloring of $\NCR_r$ does not (automatically) imply a correct $d$-defective coloring algorithm. 

Due to Lemma \ref{lemma:algDefectiveCorrespondence} a lower bound $\chi^d(\NH_1(m,\Delta)>c$ on the $d$-defective chromatic number implies that there is no one-round color reduction algorithm in the \LOCAL\ model which can (correctly)  $d$-defectively $c$-color all initially $m$-colored graphs with maximum degree $\Delta$.

Within this section we extend Definition \ref{def:source} to account for defective colorings.
\begin{definitions}[$d$-source]
	Let $d\geq 0$ be an integer, $I\subseteq V\big({\NH_1}\big)$ and $W\subseteq V\big(\NH_0(m,\Delta)\big)$.
	
	We call $x\in \NH_0(m,\Delta)$ a \emph{$(d, W)$-source} of $I$ if 
	\begin{align}
	\forall B\subseteq \Gamma(x)\cap W, ~|B|\leq d+1 ~\exists (x,A)\in I\text{ with } B\subseteq A.
	\end{align}
	\end{definitions}
	A $0$-source corresponds to Definition  \ref{def:source}. For a set $I\subseteq V\big({\NH_1}\big)$ define
	\[S^d_W(I):=\{x\in W \mid x \text{ is } (d,W)\text{-source of }I\}\]
and $S^d(I):=S_{\NH_0}^d(I)$.
		\begin{lemma}
	\label{lemma:sourcePropagation}
	For a set $I\subseteq V\big({\NH_1}\big)$ with  $\Delta(\NH_1[I])\leq d$ and $W\subseteq \NH_0$ the subgraph of $\NH_0$ induced by 
	$S^d_W(I)$ has maximum degree $d$.
\end{lemma}
\begin{proof}
Assume that the graph induced by $S^d_W(I)$ has degree at least $d+1$, i.e., there are $d+2$ distinct $(d,W)$-sources $x_0,\ldots,x_{d+1}$ in $W$ (note that $W\subseteq \NH_0$ is always a clique). For any $i=0,\ldots,d+1$ there is a node $x'_i:=(x_i,A_i)\in I$ with 
$\{x_0,\ldots,x_{d+1}\}\setminus \{x_i\}\subseteq A_i$ because $x_i$ is a $(d,W)$-source. Hence $x'_0$ is a neighbor of $x'_1, \ldots, x'_{d+1}$ in the graph $\NH_1[I]$, i.e., its degree is $d+1$, a contradiction.
\end{proof}

\begin{proof}[Proof of Theorem \ref{thm:DefectiveoneRound}]
 Assume for contradiction that we are given a $d$-defective  coloring of
  $\NH_1(m,\Delta)$ that uses at most $c= \frac{\Delta^2}{4(d+1)^2}$ colors. For
  each of the colors $k\in [c]$, the graph induced by the nodes $I_k\subseteq \NH_1$ colored with color
  $k$ has maximum degree $d$.

  Let $S\:=\bigcup_{k=1}^cS^d(I_k)$ be the set of nodes of $\NH_0=K_m$, which are
  a $(d,V(\NH_0))$-source of some color class. Note that due to Lemma \ref{lemma:sourcePropagation} $|S^d_W(I_k)|\leq d+1$ holds for each $k\in[c]$ and therefore $|S|\leq c(d+1)$. For the remainder of
  the proof, we restrict our attention to integers in
  $\bar{S}=[m]\setminus S$.
  We first fix an arbitrary set $T\subseteq \bar{S}$ of size
  $|T|=\lfloor \Delta/2 \rfloor +1$. Because $m\geq
  \frac{\Delta^2}{4(d+1)} + \Delta/2 + 1$, such a set $T$ exists. Due to Lemma \ref{lemma:sourcePropagation} and because $\NH_0[T]$ is a clique each color class can have at most $(d+1)$ distinct $(d,T)$-sources in $T$. By the pigeonhole principle there exists an $x\in T$ such that
  $x$ is a $(d,T)$-source for at most $c(d+1)/|T|$
  orientations. W.l.o.g., assume that $x\in T$ is a $(d,T)$-source 
  of color classes $I_1,\dots,I_q$, where $q\leq c(d+1)/|T|$.

  We now construct a node $(x,A)\in V\big({\NH_1}\big)$ that is not
  contained in any of the $c$ color classes $I_1,\dots,I_c$. We start by
  adding all $\lfloor \Delta/2\rfloor$ elements of $T\setminus
  \set{x}$ to $A$. Because $x$ is a $(d,T)$-source only of $I_1,\ldots, I_q$, none of the remaining $c-q$
  color classes can contain $(x,A)$. We have to add additional elements
  to $A$ in order to make sure that the color classes $I_1,\dots,I_q$
  do not contain $(x,A)$. As $T$ only consists of elements that are
  not $(d,V(\NH_0))$-sources of any of the color classes, for each color class $I_k$, $k\in[c]$, there is a set $B_k\subseteq[m]$, $|B_k|\leq d+1$ such that $(x,A)\notin I_k$ whenever $B_k\subseteq A$. For each of
  the color classes $I_k\in\set{I_1,\dots,I_q}$, we pick such a
  set $B_k$ and add $B_k$ to $A$. In this way, we obtain a pair
  $(x,A)$ that is not contained in any of the color classes
  $I_1,\dots,I_q$. The size of $A$ is
  \[  |A| \leq |T|-1 + q(d+1) \leq \left\lfloor \frac{\Delta}{2}\right\rfloor + 
  (d+1)^2\frac{c}{|T|}  < \frac{\Delta}{2} + \frac{\Delta^2/4}{\Delta/2} =
  \Delta  \]
  and thus, $(x,A)$ is a node of $\NH_1(m,\Delta)$ which is not contained in any color class, a contradiction.
\end{proof}
\clearpage
\section{Upper Bound in the \boldmath\SETLOCAL\ Model}
\label{sec:UpperBounds}
\subsubsection*{Linial's \boldmath$\bigO(\Delta^2)$-Coloring Algorithm in $\bigO(\logstar m)$ Rounds}
Linial's seminal $\bigO(\Delta^2)$-coloring algorithm (\cite{linial92}) which takes $\bigO(\logstar m)$ rounds does not require unique IDs but only an initial $m$-coloring. Let us quickly recall the algorithm to see that it works in the \SETLOCAL\ model as well: Depending on its current color, each node selects a set $F_0\subseteq[5\Delta^2\log m]$ of potential colors. After a single round of communication it learns the potential color sets $F_1,\ldots,F_{\Delta}$ of its neighbors. The sets are selected such that the set of non conflicting colors 
\begin{align*}
F_0\setminus \bigcup_{i=1}^{\Delta}F_i
\end{align*}
 is nonempty. An assignment of colors to sets in the range $[5\Delta^2\log m]$ with that property exists due to a purely combinatorial result by Erd\H{o}s et al. \cite{erdos82}. After a single round of communication each node selects one of its non-conflicting colors which yields a $5\Delta^2\log m$-coloring. The process is repeated for $\bigO(\logstar m)$ rounds (each time with a smaller $m$) until we obtain a $\bigO(\Delta^2\log\Delta)$-coloring. From there, a single additional round of the same kind suffices to directly get to $\bigO(\Delta^2)$ colors \cite{linial92}.

The \SETLOCAL\ model differs from the \LOCAL\ model only in terms of communication, that is, if two neighbors of a node $v\in V$ send the same message to $v$, $v$ will receive this message only a single time. In Linial's algorithm, to select a new color a node only needs to know the set of potential colors of its neighbors; in particular, a node's output does not change if two or more neighbors selected the same potential color set. Thus Linial's algorithm works without any modification in the \SETLOCAL\ model.

\subsubsection*{Kuhn-Wattenhofer Color Reduction Scheme}
The  color reduction scheme from \cite{Kuhn2006On} reduces an initial $m$-coloring to an $\big\lceil m\big(1-\frac{1}{\Delta+2}\big)\big\rceil$-coloring in a single communication round. Combining this with Linial's algorithm one obtains a $(\Delta+1)$-coloring algorithm with time complexity $\bigO(\Delta\log \Delta+\logstar(m))$. 
Let us take a look at a single round of the color reduction scheme as described in \cite{Kuhn2006On}: Assume that an $m$-coloring  is given and let $q$ be the desired number of colors of a new coloring. All nodes $v$ with a color smaller or equal to $q$ keep their color. Only nodes with one of the colors $q+1,\ldots,m$ need to choose a color which is smaller or equal to $q$. We number those colors from $x_0,\ldots, x_{t-1}$ where $t=m-q$. Then a node with color $x_i$ selects a new color from the range $R_i=\{i(\Delta+1)+1,\ldots,(i+1)(\Delta+1)\}$ which is different from the initial color of its neighbors. To actually obtain a $q$-coloring each color in the range $R_i$ needs to be smaller or equal to $q$, which implies the condition $q\geq m\big(1-\frac{1}{\Delta+2}\big)$.
 Nodes with different colors choose their colors from disjoint ranges; hence the obtained $q$-coloring is feasible and all nodes with a color greater than $q$ can choose their color at the same time, i.e., only one round of communication is needed.

In the above algorithm a node only needs to know its own color to determine its range $R_i$ and the set of colors of its neighbors to actually select one of the colors from this range. Its output does not depend on how many neighbors have a certain color, it is sufficient to know whether a color was selected by a neighbor at all. Hence the color reduction scheme works in the \SETLOCAL\ model without any modification.
\begin{proof}[Proof of Theorem \ref{thm:upperbound}]
The proof follows with the above arguments.
\end{proof}

\clearpage

\bibliography{references}

\newcommand{\etalchar}[1]{$^{#1}$}
\begin{thebibliography}{GKK{\etalchar{+}}07}

\bibitem[ABI86]{alon86}
N.~Alon, L.~Babai, and A.~Itai.
\newblock A fast and simple randomized parallel algorithm for the maximal
  independent set problem.
\newblock {\em J.\ of Algorithms}, 7(4):567--583, 1986.

\bibitem[AGLP89]{awerbuch89}
B.~Awerbuch, A.~V. Goldberg, M.~Luby, and S.~A. Plotkin.
\newblock Network decomposition and locality in distributed computation.
\newblock In {\em Proc. 30th Symp. on Found.\ of Computer Science (FOCS)},
  pages 364--369, 1989.

\bibitem[Alo10]{alon10}
N.~Alon.
\newblock On constant time approximation of parameters of bounded degree
  graphs.
\newblock In {\em Proc.\ Workshop on Property Testing}, pages 234--239, 2010.

\bibitem[Bar15]{barenboim15}
L.~Barenboim.
\newblock Deterministic ({$\Delta$} + 1)-coloring in sublinear (in {$\Delta$})
  time in static, dynamic and faulty networks.
\newblock In {\em Proc.~34th ACM Symposium on Principles of Distributed
  Computing (PODC)}, pages 345--354, 2015.

\bibitem[BE10a]{barenboim10}
L.~Barenboim and M.~Elkin.
\newblock Deterministic distributed vertex coloring in polylogarithmic time.
\newblock In {\em Proc.\ 29th Symp.\ on Principles of Distributed Computing
  (PODC)}, 2010.

\bibitem[BE10b]{Barenboim2010}
Leonid Barenboim and Michael Elkin.
\newblock Sublogarithmic distributed mis algorithm for sparse graphs using
  nash-williams decomposition.
\newblock {\em Distributed Computing}, 22(5):363--379, 2010.

\bibitem[BE13]{barenboimelkin_book}
L.~Barenboim and M.~Elkin.
\newblock {\em Distributed Graph Coloring: Fundamentals and Recent
  Developments}.
\newblock Morgan \& Claypool Publishers, 2013.

\bibitem[BEK15]{BEK15}
L.~Barenboim, M.~Elkin, and F.~Kuhn.
\newblock Distributed ({Delta}+1)-coloring in linear (in {Delta}) time.
\newblock {\em SIAM J.~Computing}, 43(1):72--95, 2015.

\bibitem[BEPS12]{barenboim12}
L.~Barenboim, M.~Elkin, S.~Pettie, and J.~Schneider.
\newblock The locality of distributed symmetry breaking.
\newblock In {\em Proc.\ 53th Symp.\ on Foundations of Computer Science
  (FOCS)}, 2012.

\bibitem[BFH{\etalchar{+}}16]{LLL_lowerbound}
Sebastian Brand, Orr Fischer, Juho Hirvonen, Barbara Keller, Tuomo
  Lempi\"ainen, Joel Rybicki, Jukka Suomela, and Jara Uitto.
\newblock A lower bound for the distributed lov\'asz local lemma.
\newblock In {\em Proc.~48th Symp.~on the Theory of Computing (STOC)}, 2016.

\bibitem[CKP16]{chang16}
Y.-J. Chang, T.~Kopelowitz, and S.~Pettie.
\newblock An exponential separation between randomized and deterministic
  complexity in the {LOCAL} model.
\newblock {\em CoRR}, abs/1602.08166, 2016.

\bibitem[CV86]{cole86}
R.~Cole and U.~Vishkin.
\newblock Deterministic coin tossing with applications to optimal parallel list
  ranking.
\newblock {\em Information and Control}, 70(1):32--53, 1986.

\bibitem[EFF82]{erdos82}
P.~Erd\H{o}s, P.~Frankl, and Z.~F{\"u}redi.
\newblock Families of finite sets in which no set is covered by the union of
  two others.
\newblock {\em Journal of Combinatorial Theory, Series A}, 33(2):158 -- 166,
  1982.

\bibitem[FHK15]{fraigniaud15}
P.~Fraigniaud, M.~Heinrich, and A.~Kosowski.
\newblock Local conflict coloring.
\newblock {\em CoRR}, abs/1511.01287, 2015.

\bibitem[GHS14]{goeoes14_PODC}
M.~G\"o\"os, J.~Hirvonen, and J.~Suomela.
\newblock Linear-in-{Delta} lower bounds in the {LOCAL} model.
\newblock In {\em Proc.~33rd ACM Symposium on Principles of Distributed
  Computing (PODC)}, pages 86--95, 2014.

\bibitem[GKK{\etalchar{+}}07]{greedycoloring}
C.~Gavoille, R.~Klasing, A.~Kosowski, \L{}. Kuszner, and A.~Navarra.
\newblock On the complexity of distributed graph coloring with local minimality
  constraints.
\newblock Technical Report 6399, INRIA, 2007.

\bibitem[GPS88]{goldberg88}
A.V. Goldberg, S.A. Plotkin, and G.E. Shannon.
\newblock Parallel symmetry-breaking in sparse graphs.
\newblock {\em SIAM Journal on Discrete Mathematics}, 1(4):434--446, 1988.

\bibitem[GS14]{goeoes14_DISTCOMP}
M.~G\"o\"os and J.~Suomela.
\newblock No sublogarithmic-time approximation scheme for bipartite vertex
  cover.
\newblock {\em Distributed Computing}, 27(6):435--443, 2014.

\bibitem[HJK{\etalchar{+}}15]{hella15}
L.~Hella, M~J\"arvisalo, A.~Kuusisto, J.~Laurinharju, T.~Lampi\"ainen,
  K.~Luosto, J.~Suomela, and J.~Virtema.
\newblock Weak models of distributed computing, with connections to modal
  logic.
\newblock {\em Distributed Computing}, 28(1):31--53, 2015.

\bibitem[HSS16]{hsinhao_coloring}
S.~G. Harris, J.~Schneider, and H.-H. Su.
\newblock Distributed ({$\Delta+1$})-coloring in sublogarithmic rounds.
\newblock In {\em Proc.~48th Symp.~on the Theory of Computing (STOC)}, 2016.

\bibitem[Kar72]{Karp1972}
Richard~M. Karp.
\newblock Reducibility among combinatorial problems.
\newblock In {\em Proc.~Symp. on Complexity of Computer Computations}, pages
  85--103, 1972.

\bibitem[KMW04]{lowerbound}
F.~Kuhn, T.~Moscibroda, and R.~Wattenhofer.
\newblock What cannot be computed locally!
\newblock In {\em Proc. 23rd ACM Symp. on Principles of Distributed Computing
  (PODC)}, pages 300--309, 2004.

\bibitem[KMW16]{kuhn16_jacm}
F.~Kuhn, T.~Moscibroda, and R.~Wattenhofer.
\newblock Local computation: Lower and upper bounds.
\newblock {\em J.~of the ACM}, 63(2), 2016.

\bibitem[Kuh09]{stacs09}
F.~Kuhn.
\newblock Local multicoloring algorithms: Computing a nearly-optimal {TDMA}
  schedule in constant time.
\newblock In {\em Proc.\ of Symp.\ on Theoretical Aspects of Computer Science
  (STACS)}, pages 613--624, 2009.

\bibitem[KW06]{Kuhn2006On}
F.~Kuhn and R.~Wattenhofer.
\newblock On the complexity of distributed graph coloring.
\newblock In {\em Proc.~25th ACM Symposium on Principles of Distributed
  Computing (PODC)}, pages 7--15, 2006.

\bibitem[Lin92]{linial92}
N.~Linial.
\newblock Locality in distributed graph algorithms.
\newblock {\em SIAM Journal on Computing}, 21(1):193--201, 1992.

\bibitem[Lub86]{luby86}
M.~Luby.
\newblock A simple parallel algorithm for the maximal independent set problem.
\newblock {\em SIAM Journal on Computing}, 15:1036--1053, 1986.

\bibitem[Pel00]{peleg00}
D.~Peleg.
\newblock {\em Distributed Computing: A Locality-Sensitive Approach}.
\newblock SIAM, 2000.

\bibitem[PS95]{panconesi95}
A.~Panconesi and A.~Srinivasan.
\newblock On the complexity of distributed network decomposition.
\newblock {\em Journal of Algorithms}, 20(2):581--592, 1995.

\bibitem[SV93]{szegedy93}
M.~Szegedy and S.~Vishwanathan.
\newblock Locality based graph coloring.
\newblock In {\em Proc. of the 25th ACM Symposium on Theory of Computing
  (STOC)}, pages 201--207, 1993.

\bibitem[Zuc07]{Zuckerman07}
David Zuckerman.
\newblock Linear degree extractors and the inapproximability of max clique and
  chromatic number.
\newblock {\em Theory of Computing}, 3(1):103--128, 2007.

\end{thebibliography}

\end{document}